\documentclass[11pt]{article}%
\usepackage{amsfonts}
\usepackage{amsmath}
\usepackage{amssymb}
\usepackage{graphicx}%
\providecommand{\U}[1]{\protect\rule{.1in}{.1in}}
\newtheorem{theorem}{Theorem}
\newtheorem{acknowledgement}[theorem]{Acknowledgement}

\newtheorem{corollary}[theorem]{Corollary}

\newtheorem{definition}[theorem]{Definition}

\newtheorem{proposition}[theorem]{Proposition}

\newenvironment{proof}[1][Proof]{\noindent\textbf{#1.} }{\ \rule{0.5em}{0.5em}}
\begin{document}

\title{Polar Duality and Quasi-States: :a Geometric Picture of Quantum Indeterminacy}
\author{Maurice de Gosson\thanks{maurice.de.gosson@univie.ac.at}\\Faculty of Mathematics, NuHAG\\University of Vienna\\and\\ARI - \ Austrian Academy of Sciences}
\maketitle

\begin{abstract}
The aim of this paper is to suggest a new interpretation \ of quantum
indeterminacy using the notion of polar duality from convex geometry.. i.e.
Our approach does not involve the usual variances and covariances, whose use
to describe quantum uncertainties has been questioned by Uffink and
Hilgevoord. We introduce the geometric notion of "quasi-states" (which we
could also have called Cheshire cat states) which are related in a way that
will be explained, to the notion of "quantum blob" we have introduced in
previous work. The consideration of the symmetries of the quasi-states leads
to the definition of the canonical group of a quasi state, which allows to
clssify them.

\end{abstract}

\section{Prologue and Introduction}

Mathematical points do not have any physical meaning: that I was already
taught a long time ago by a primary school teacher (whose name I have
unfortunately forgotten). However, mathematical points, serve as fundamental
elements in classical physics, facilitating a continuous description of
phenomena through analysis and geometry, which are deeply rooted in Newtonian
principles \cite{ICP}. These abstract entities underpin classical phase space,
crucial for understanding physical systems. However, transition to the quantum
realm introduces conceptual challenges due to Heisenberg's uncertainty
principle, rendering the notion of points obsolete as particles can no longer
be precisely localized in phase space. The uncertainty principle of quantum
mechanics (which we prefer to call the principle of indeterminacy) is
conventionally expressed in terms of variances and covariances of the involved
observables (for instance position and momentum). This characterization of
uncertainties is arbitrary and was already criticized by Hilgevoord and Uffink
\cite{hi,hiuf}, who showed variances and covariances are satisfactory
statistical measures of uncertainties only for Gaussian (or close to Gaussian)
distributions. Moreover, several authors, foremost Karl Popper (1967), have
contested the view that the Heisenberg uncertainty principle could be granted
the status of a true principle on the grounds that they are derivable from the
theory, whereas one cannot obtain the theory from the uncertainty relations.
The argument was that one can never derive any equation, say, the
Schr\"{o}dinger equation, or the commutation relation, from an in
inequality\footnote{We do not however totally agree with this statement; may
mathematical theories (e.g. functional spaces) are based on integral
inequalities.} (see the discussion in Stanford Encyclopedia article
\cite{Stanford}). These conceptual problems, together with the search for a
substitute of phase space in quantum mechanics has led us has led us to
suggest in \cite{gopolar,sym,FOOPMC} a new more general geometric expression
of the principle of indeterminacy using the notion of polar duality familiar
from convex geometry. It turns out that this reformulation allows to recover
the traditional Heisenberg and Robertson--Schr\"{o}dinger inequalities. This
motivates a posteriori our approach, which consists in introducing a kind of
\textquotedblleft quantum Cheshire cat state\textquotedblright, which we call
\textquotedblleft quasi state\textquotedblright. Traditionally, such states
refer to a situation where a quantum particle and one of its properties appear
to be separated in space. The term is derived from the Cheshire Cat in Lewis
Carroll's \textit{Alice's Adventures in Wonderland}, which can make its body
disappear while leaving its grin behind. In our contexts in the simplest case,
the \textquotedblleft grin is the product .$X\times X^{\hbar}$ where $X$ is a
compact subset of configuration space and $X^{\hbar}$ its polar dual in
momentum space. These products appear as characterizing minimum uncertainty
states in position and momentum,, and are generalized by applying symplectic
transformations to them. \ This will which is briefly described in Section
\ref{secpolar} where we relate it to the notion of quantum blob which is
central to the present paper. This led us to introduce and study
\cite{blob,physletta,golu09} the notion of \textquotedblleft quantum
blobs\textquotedblright\ as being the smallest phase space units of phase
space compatible with the uncertainty principle and having the symplectic
group as group of symmetries. Quantum blobs are in a bijective correspondence
with the squeezed coherent states from standard quantum mechanics, of which
they are a phase space picture.

\subsection*{Notation and terminology}

We collect here in a concise way the main mathematical objects that will be
used. General reference texts for the notions of symplectic geometry used here
see \cite{Silva,Birk}.

\paragraph{The symplectic phase space}

The cotangent bundle $T^{\ast}\mathbb{R}^{n}=\mathbb{R}^{n}\times
(\mathbb{R}^{n})^{\ast}$ where $\mathbb{R}^{n}$ is viewed as \textquotedblleft
configuration space\textquotedblright\ $\mathbb{R}_{x}^{n}$\ is denoted by
$\mathbb{R}_{x}^{n}\times\mathbb{R}_{p}^{n}$ \ and will be most of the time
identified with $\mathbb{R}^{2n}$. It will be equipped with the standard
symplectic form $\sigma=dp\wedge dx=d(pdx)$ where $pdx=p_{1}dx_{1}+\cdot
\cdot\cdot+p_{n}dx_{x}$ is the canonical (tautological) $1$-form on $T^{\ast
}\mathbb{R}^{n}$). In matrix form $\sigma(z,z^{\prime})=Jz\cdot z^{\prime}$
for $z,z^{\prime}\in\mathbb{R}^{2n}$ where $J=%
\begin{pmatrix}
0_{n\times n} & I_{n\times n}\\
-I_{n\times n} & 0_{n\times n}%
\end{pmatrix}
$ and $\cdot$ is the standard Euclidean scalar product on $\mathbb{R}^{2n}$.

\paragraph{Symplectic group}

The symplectic group \cite{Birk} $\operatorname*{Sp}(n)$ is the group of all
linear automorphisms of $\mathbb{R}^{2n}=T^{\ast}\mathbb{R}^{n}$ preserving
the canonical symplectic 2-form $\sigma=dp\wedge dx$. in coordinates:
$S\in\operatorname*{Sp}(n)$ if and only $S^{T}JS=SJS^{T}=J$. .The symplectic
group is generate by $J$ and the automorphisms $M_{L}:(x,p)\longmapsto
(L^{-1}x,L^{T}p)$ ($L\in GL(n,\mathbb{R})$) and $V_{-P}:(x,p)\longmapsto
(x,Px+p)$ ($P\in$,$\operatorname*{Sym}(n,\mathbb{R}))$, the additive group of
real symmetric n$\times n$ matrices).

\paragraph{Lagrangian Grassmannian}

We denote by $\operatorname*{Lag}(n)$ the Lagrangian Grassmannian of the
symplectic space $(\mathbb{R}^{2n},\sigma):$ we have $\ell\in
\operatorname*{Lag}(n)$ if and only if $\ell$ is a linear subspace of
$\mathbb{R}^{2n}$ such that $\dim\ell=n$ and the restriction of $\sigma$ to
$\ell$ is zero $\operatorname*{Sp}(n)$ acts transitively on
$\operatorname*{Lag}(n)$. We will use the notation $\ell_{X}=\mathbb{R}%
_{x}^{n}\times0$ and $\ell_{P}=0\times\mathbb{R}_{p}^{n}$ (\textquotedblleft
coordinate Lagrangian planes\textquotedblright).

\paragraph{Metaplectic group}

The metaplectic group $\operatorname*{Mp}(n)$ is the unitary representation of
the double cover of $\operatorname*{Sp}(n)$ in $L^{2}(\mathbb{R}^{n})$
\cite{Birk}. It is generated by the unitary operators $\widehat{J}$,
$\widehat{M}_{L,m}$, and $\widehat{V}_{-P}$ defined for $\psi\in
L^{2}(\mathbb{R}^{n})$ by%
\begin{gather}
\widehat{J}\psi(x)=\left(  \tfrac{1}{2\pi i\hbar}\right)  ^{n/2}%
\int_{\mathbb{R}^{n}}e^{-\frac{i}{\hbar}x\cdot x^{\prime}}\psi(x^{\prime
})dx^{\prime}.\label{mp1}\\
\widehat{M}_{L,m}\psi(x)=i^{m}\sqrt{|\det L|}\psi(Lx)\text{\ },\text{ \ }%
m\pi=\arg\det L,\label{mp2}\\
\widehat{V}_{P}\psi(x)=e^{-\frac{i}{2\hbar}Px\cdot x}\psi(x)\text{ }
\label{mp3}%
\end{gather}
and which cover $J,M_{L},V_{P}$, respectively. The integer $m$ (the
\textquotedblleft Maslov index\textquotedblright) is uniquely determined
modulo $4$.

\paragraph{The displacement operator}

The Heisenberg--Weyl displacement operator $\widehat{T}(z_{0})$ ($z_{0}%
=(x_{0},p_{0})$) is the unitary mapping $L^{2}(\mathbb{R}^{n})\longrightarrow
L^{2}(\mathbb{R}^{n})$ defined by
\[
\widehat{T}(z_{0})\psi(x)=e^{\frac{i}{\hbar}(p_{0}\cdot x-\frac{1}{2}%
p_{0}\cdot x_{0})}\psi(x-x_{0}).
\]
It satisfies the symplectic covariance property: if $\widehat{S}%
\in\operatorname*{Mp}(n)$ covers $S\in\operatorname*{Sp}(n($ then%
\begin{equation}
\widehat{S}\widehat{T}(z_{0})=\widehat{T}(Sz_{0})\widehat{S}. \label{Tcov}%
\end{equation}

\paragraph{The Wigner transform}

The Wigner transform of $\psi\in L^{2}(\mathbb{R}^{n})$ is defined by
\begin{equation}
W\psi(z)=(\pi\hbar)^{-n}(\widehat{T}(-z)\psi|R\widehat{T}(-z)\psi
)_{L^{2}(\mathbb{R}^{n})} \label{wig0}%
\end{equation}
where $R$ is the reflection operator: $R\psi(x)=\psi(-x)$; in explicit form%
\begin{equation}
W\psi(x,p)=\left(  \tfrac{1}{2\pi\hbar}\right)  ^{n}\int_{\mathbb{R}^{n}%
}e^{-\frac{i}{\hbar}p\cdot y}\psi(x+\tfrac{1}{2}y)\overline{\psi(x-\tfrac
{1}{2}y)}dy. \label{wig1}%
\end{equation}
It follows from (\ref{wig0}) and (\ref{Tcov}) that%
\begin{equation}
W(\widehat{S}\psi)(z)=W\psi(S^{-1}z) \label{Wcov}%
\end{equation}
(symplectic covariance of the Wigner transform).

\section{Polar Duality and Quasi States\label{secpolar}}

In previous work \cite{Bull,sym,FOOPMC} we have used the geometric notion of
polar duality (and generalizations thereof) to propose a general formulation
of the quantum indeterminacy; in \cite{IOP} we have applied our constructions
to Pauli's reconstruction problem.

\subsection{Polar duality}

Let $X\subset\mathbb{R}_{x}^{n}$ be a symmetric convex body, that is, $X=-X$
is compact\ and convex, and has non-empty interior. The $\hbar$-polar dual of
$X$ is, by definition,n the symmetric convex body%
\begin{equation}
X^{\hslash}=\{p\in\mathbb{R}^{n}:\sup\nolimits_{x\in X}(p\cdot x)\leq\hbar\}.
\label{xh}%
\end{equation}
The following reflexivity and anti-monotonicity properties of polar duality
are obvious:
\begin{equation}
(X^{\hslash})^{\hbar}=X\text{ \ },\text{ \ }\ X\subset Y\Longrightarrow
Y^{\hslash}\subset X^{\hslash} \label{monoref}%
\end{equation}
and, for all $L\in GL(n,\mathbb{R})$,
\begin{equation}
(LX)^{\hbar}=(L^{T})^{-1}X^{\hslash}. \label{scaling}%
\end{equation}
Let $B_{X}^{n}(\sqrt{\hbar})$ (\textit{resp}. $B_{P}^{n}(\sqrt{\hbar})$) be
the centered ball with radius $\sqrt{\hbar}$ in $\mathbb{R}_{x}^{n}$
(\textit{resp}. $\mathbb{R}_{p}^{n}$). We have the
%TCIMACRO{\U{b4}}%
%BeginExpansion
\'{}%
%EndExpansion
self-duality property
\begin{equation}
(B_{X}^{n}(\sqrt{\hbar}))^{\hbar}=B_{P}^{n}(\sqrt{\hbar}). \label{balls}%
\end{equation}

Let us view the symmetric convex body $X\subset\mathbb{R}_{x}^{n}$ as
representing a cloud of position measurements on the quantum system under
investigation. In \cite{FOOPMC} we postulated that the polar dual $X^{\hslash
}$ consists of the set of outcomes of possible simultaneous momentum measurements.

Lagrangian planes are \ $n$-dimensional subspaces of $\mathbb{R}^{2n}$ on
which the symplectic form vanishes identically. In the case $n=1$ the
symplectic form is (up to the sign) the determinant function, hence the
Lagrangian planes in the phase plane are just the straight lines passing
through the origin. I the case \ of arbitrary $n$ Lagrangian planes are
typically the planes of coordinates $(x_{j\in\Lambda},p_{k\in\Lambda^{o}})$
where $\Lambda,\Lambda^{o}$ is a partition of \ $\{1,...,n\}$.

We extend the results above to more general situations. We recall \cite{Birk}
that the action of $\operatorname*{Sp}(n)$ on the Lagrangian Grassmannian
$\operatorname*{Lag}(n)$ is transitive, and so is its action on pairs of
transverse Lagrangian planes. In particular each $\ell\in\operatorname*{Lag}%
(n)$ ) is the image of $\ell_{X}=\mathbb{R}_{x}^{n}\times0$ (\textit{resp.}
$\ell_{P}=0\times\mathbb{R}_{p}^{n}$) by some $S\in\operatorname*{Sp}(n)$, and
each pair $(\ell,\ell^{\prime})$ with $\ell\cap\ell^{\prime}=0$ is the image
by some $S\in\operatorname*{Sp}(n)$ of the canonical pair $(\ell_{X},\ell
_{P})$.

\begin{definition}
A Lagrangian frame is the image of the canonical position-momentum frame
$(\ell_{X},\ell_{P})$ by some $S\in\operatorname*{Sp}(n)$: $(\ell,\ell
^{\prime})=S(\ell_{X},\ell_{P})$; equivalently it is a pair $(\ell
,\ell^{\prime})$ of transversal Lagrangian planes. The \ set of Lagrangian
frames in the symplectic space $(\mathbb{R}^{2n},\sigma)$ is denoted
$\mathcal{LF}(n)$ .
\end{definition}

We thus have a transitive action
\[
\operatorname*{Sp}(n)\times\mathcal{LF}(n)\longrightarrow\mathcal{LF}%
(n)\text{.}%
\]

Let $((\ell,\ell^{\prime})\in\mathcal{LF}(n)$ and $X_{\ell}\subset\ell$ be a
centrally symmetric convex body. Let $S\in\operatorname*{Sp}(n)$ be such that
$((\ell,\ell^{\prime})=S(\ell_{X},\ell_{P})$ and define $X\subset
\mathbb{R}_{x}^{n}$ (which we identify with $X\times0\subset\ell_{X}%
\subset\mathbb{R}_{x}^{n}\times0$) by $X_{\ell}=S(X\times0)$.

\begin{definition}
The symplectic polar dual of $X_{\ell}\subset\ell$ with respect to
$\ell^{\prime}$ is $X_{\ell^{\prime}}^{\hbar}=S(0\times X^{\hbar})$ where
$X^{\hbar}\subset\ell_{P}=0\times\mathbb{R}_{p}^{n}$ is the usual polar dual
of $X$.
\end{definition}

The definition of $X_{\ell^{\prime}}^{\hbar}$ does not depend on the choice of
$S$. In fact, if $((\ell,\ell^{\prime})=S^{\prime}(\ell_{X},\ell_{P})$ then
$S^{-1}S^{\prime}(\ell_{X},\ell_{P})=(\ell_{X},\ell_{P})$ so that we must have
$S^{\prime}=SM_{L}$ where $M_{L}=%
\begin{pmatrix}
L^{-1} & 0\\
0 & L^{T}%
\end{pmatrix}
$. Suppose that $S^{\prime}(X^{\prime}\times0)=S(X\times0)$, then
$M_{L}(X^{\prime}\times0)=X\times0)$ that is $(L^{-1}X^{\prime}\times
0)=X\times0)\Longleftrightarrow X^{\prime}=LX$ hence $(X^{\prime})^{\hbar
}=(L^{T})^{-1}X^{\hbar}$; it follows that
\begin{align*}
S^{\prime}(0\times(X^{\prime})^{\hbar})  &  =S^{\prime}(0\times(L^{T}%
)^{-1}X^{\hbar})\\
&  =S^{\prime}M_{L}^{-1}((0\times X^{\hbar})=S((0\times X^{\hbar}).
\end{align*}

The symplectic polar dual can be defined intrinsically:\ it is easy to verify
\cite{Bull} that\ $X_{\ell^{\prime}}^{\hbar}$ is the subset of $\ell^{\prime}$
consisting of all $z^{\prime}\in\ell^{\prime}$ such that
\begin{equation}
X_{\ell^{\prime}}^{\hbar}=\{z^{\prime}\in\ell^{\prime}:\sigma(z^{\prime
},z)\leq\hbar\text{ \ \textit{for all} \ }z\in X_{\ell}\}. \label{ozz}%
\end{equation}

\subsection{Quasi states: definition and first properties}

The considerations above suggest that we consider pairs $(X_{\ell}%
,X_{\ell^{\prime}}^{\hbar})$ representing in some way a quantum state. We will
call such a representation an \textquotedblleft quasi state\textquotedblright,
which has similarities (but also a different meaning) with the Cheshire cat
states appearing in the theory of weak measurements. \ In the latter case
there has bee criticism of the underlying idea that claiming separation of a
property from a particle is preposterous, (see the discussion in Duprey et al.
\cite{Duprey}). The following definition will be justified in a moment:

\begin{definition}
Let $(\ell,\ell^{\prime})=S^{\prime}(\ell_{X},\ell_{P})\in\mathcal{LF}(n)$ and
$X_{\ell}\subset\ell$ be a ellipsoid. The subset of the phase space
$\mathbb{R}^{2n}$ defined by
\[
X_{\ell}\times X_{\ell^{\prime}}^{\hbar}=S(X\times X^{\hbar})=S(X)\times
S(X^{\hbar})
\]
is called a minimum uncertainty quasi state.
\end{definition}

The key property of $X_{\ell}\times X_{\ell^{\prime}}^{\hbar}$, which will
allow us to connect this definition with the standard theory of minimum
uncertainty states expressed in terms of covariances is the following:

\begin{theorem}
Assume that $X$ is an ellipsoid: $X=L^{-1}(B_{X}^{n}(\sqrt{\hbar}))$ for some
$L\in GL(n,\mathbb{R})$. The John ellipsoid of the minimum uncertainty quasi
state $\ X_{\ell}\times X_{\ell^{\prime}}^{\hbar}$ is the quantum blob
$Q=SM_{L}(B^{2n}(\sqrt{\hbar}))$ and we have
\begin{equation}
\Pi_{\ell,\ell^{\prime}}Q=X_{\ell}\text{ \ , \ }\Pi_{\ell^{\prime}%
,\ell^{\prime}}Q=X_{\ell^{\prime}}^{\hbar} \label{proj1}%
\end{equation}
where \ $\Pi_{\ell,\ell^{\prime}}$ (resp. $\Pi_{\ell^{\prime},\ell^{\prime}}$)
is the projection on $\ell$ along $\ell^{\prime}$ (resp. onto $\ell^{\prime}$
along $\ell$).
\end{theorem}

\begin{proof}
Recall \cite{Ball} that the John ellipsoid of a convex set is the (unique)
ellipsoid with maximum volume contained in that set. In view of formula
(\ref{scaling}) we We have $X^{\hbar}=L^{T}((B_{P}^{n}(\sqrt{\hbar}))$ hence
\[
X\times X^{\hbar}=M_{L}(B_{X}^{n}(\sqrt{\hbar})\times B_{P}^{n}(\sqrt{\hbar
}))
\]
so that
\[
X_{\ell}\times X_{\ell^{\prime}}^{\hbar}=S\left[  M_{L}(B_{X}^{n}(\sqrt{\hbar
})\times B_{P}^{n}(\sqrt{\hbar}))\right]
\]
and it suffices to prove that the John ellipsoid \ of $B_{X}^{n}(\sqrt{\hbar
})\times B_{P}^{n}(\sqrt{\hbar})$ is $B^{2n}(\sqrt{\hbar})$, which is
straightforward to check \cite{Bull}. The formulas (\ref{proj1}) are clearly
true when $S=I$, i.e. when $X_{\ell}\times X_{\ell^{\prime}}^{\hbar}=X\times
X^{\hbar}$ since in this case $Q$ is just the ball $B^{2n}(\sqrt{\hbar})$. The
general case follows from the fact that $\Pi_{\ell,\ell^{\prime}}=S\Pi
_{X}S^{-1}$ and $\Pi_{\ell^{\prime},\ell}=S\Pi_{P}S^{-1}$ where $\Pi_{X}$ and
$\Pi_{P}$ \ are the orthogonal projections onto $\ell_{X}$ and $\ell_{P}$ , respectively.
\end{proof}

\subsection{Mahler's volune as a measure of indeterminacy}

Let $X$ be a convex body in $\mathbb{R}_{x}^{n}$. We assume that $X$ contains
$0$ in its interior. By definition, the Mahler volume \cite{Mahler} of $X$ is
the product%
\begin{equation}
\upsilon(X)=\operatorname*{Vol}(X)\operatorname*{Vol}(X^{\hbar}%
)=\operatorname*{Vol}(X\times X^{\hbar}) \label{Mahler}%
\end{equation}
where $\operatorname*{Vol}$ is the usual Euclidean volume on $\mathbb{R}%
_{x}^{n}$ or \ $\mathbb{R}_{p}^{n}$. The Mahler volume is a dimensionless
quantity because we have $\upsilon(\lambda X)=\upsilon(X)$ for all $\lambda
>0$. More generally if $L$ is an automorphism of $\mathbb{R}_{x}^{n}$ then we
have%
\begin{equation}
\upsilon(LX)=\operatorname*{Vol}(LX)\operatorname*{Vol}((L^{T})^{-1}X^{\hbar
})=\upsilon(X). \label{mahlerinv}%
\end{equation}

The following upper bound for the Mahler volume \ (Blaschke--Santal\'{o}
inequality \cite{Blaschke,Santalo}) is well-known
\begin{equation}
\upsilon(X)\leq\operatorname*{Vol}B_{X}^{n}(\sqrt{\hbar})\operatorname*{Vol}%
B_{X}^{n}(\sqrt{\hbar})^{\hbar}|=\frac{(\pi\hbar)^{n}}{\Gamma(\frac{n}%
{2}+1)^{2}} \label{volxvolxh}%
\end{equation}
with equality if and only $X$ is an ellipsoid \ $LB_{X}^{n}(\sqrt{\hbar})$.
Since symplectic transformations are volume preserving it follows that that we
have the following bound for quasi states:%
\begin{equation}
\operatorname*{Vol}(X_{\ell}\times X_{\ell^{\prime}}^{\hbar})\leq\frac
{(\pi\hbar)^{n}}{\Gamma(\frac{n}{2}+1)^{2}} \label{BSquasi}%
\end{equation}
with again equality if and only if $X$ is an ellipsoid. For the question of
lover bounds for the Mahler volume it is known that (\cite{BM,Kuper})
$\upsilon(X)\geq\frac{(\pi\hbar)^{n}}{4^{n}n!}$which implies for quasi states
\begin{equation}
\operatorname*{Vol}(X_{\ell}\times X_{\ell^{\prime}}^{\hbar}).\geq\frac
{(\pi\hbar)^{n}}{4^{n}n!}. \label{kuper}%
\end{equation}
Actually, a famous conjecture, due to Mahler himself, is that $\upsilon
(X)\geq\frac{(4\hbar)^{n}}{n!}$ leading to
\begin{equation}
\operatorname*{Vol}(X_{\ell}\times X_{\ell^{\prime}}^{\hbar})\geq\frac
{(4\pi\hbar)^{n}}{n!}. \label{Mahlerquasi}%
\end{equation}
This lower bound is actually reached when $X$ is any $n$-parallelepiped%
\begin{equation}
X=[-\sqrt{2\Delta x_{1}^{2}},\sqrt{2\sigma\Delta x_{1}^{2}}]\times\cdot
\cdot\cdot\times\lbrack-[-\sqrt{2\Delta x_{n}^{2}},\sqrt{2\sigma\Delta
x_{n}^{2}}] \label{interval}%
\end{equation}
as is easily verified by a direct computation. Note that this situation
corresponds to the saturation of the Heisenberg inequalities, that is to a
tensor product of one-dimensional coherent states.

We interpret the Mahler volume of a quasi state as a measure of indeterminacy;
intuitively tis should be clear: for instance when the frame $(\ell
,\ell^{\prime})$ is the canonical one, for a given indeterminacy $X$ in
position , the larger \ $\operatorname*{Vol}(X_{\ell}\times X_{\ell^{\prime}%
}^{\hbar})$ is the larger is the indeterminacy in momenta. Suppose indeed that
$X_{\ell}\times P_{\ell^{\prime}}$ , $P_{\ell^{\prime}}\supset X_{\ell
^{\prime}}^{\hbar}$ is an arbitrary (that is, n not necessarily minimum
uncertainty) quasi state; then (\ref{Mahlerquasi}) say that we will have%
\[
\operatorname*{Vol}(X_{\ell}\times P)\geq\frac{(4\pi\hbar)^{n}}{n!}.
\]
The Blaschke--Santal\'{o} bound the appears as a limiting case occurring when
the quasi state becomes minimum uncertainty (Gaussian states for instance),
occurring when $P=$ $X_{\ell^{\prime}}^{\hbar})$. We conjecture that the
Mahler conjecture could be proven using quantum considerations; we will come
back to this possibility in an ongoing work.

\section{\ Topological Properties of Quasi States}

\subsection{Symplectic capacities}

Let $\operatorname*{Symp}(n)$ the group of all symplectomorphisms of
$(T^{\ast}\mathbb{R}^{n},\sigma)=(\mathbb{R}^{2n},\sigma)$: $\ f\in
\operatorname*{Symp}(n)$ if and only if $f$ is a diffeomorphism of
$\mathbb{R}^{2n}$ preserving the 2-form $\sigma$. Equivalently, the Jacobian
matrix $D_{z}f(z)$ \ of $f$ is in $\operatorname*{Sp}(n)$ for every
$z\in\mathbb{R}^{2n}$. \ Clearly
\[
\operatorname*{Sp}(n)\subset\operatorname*{ISp}(n)\subset\operatorname*{Symp}%
(n).
\]
A simple example of a non-linear symplectomorphism is given by the lift $f$ of
a diffeomorphism $\phi$ of $\mathbb{R}^{n}$ to $T^{\ast}\mathbb{R}^{n}$:
\[
f(z)=(\phi(x),(D\phi(x)^{-1})^{T}p).
\]
Symplectomorphisms not only preserve phase-space volumes, but they also
preserve the more subtle notion of symplectic capacity. A (normalized)
symplectic capacity \cite{Ekeland,HZ} is a function associating to every
subset $\Omega$ of the symplectic phase space $(\mathbb{R}^{2n},\sigma)$ a
number $c(\Omega)\in\lbrack0,\infty]$ and satisfying the following axioms:

\textit{Monotonicity}: If $\Omega\subset\Omega^{\prime}$ then $c(\Omega)\leq
c(\Omega^{\prime})$;

\textit{Conformality}: For every $\lambda\in\mathbb{R}$ we have $c(\lambda
\Omega)=\lambda^{2}c(\Omega)$;

\textit{Symplectic invariance}: $c(f(\Omega))=c(\Omega)$ for every
$f\in\operatorname*{Symp}(n)$;

\textit{Normalization}: We have, for $1\leq j\leq n$,
\begin{equation}
c(B^{2n}(R))=\pi R^{2}=c(Z_{j}^{2n}(R)) \label{cbz}%
\end{equation}
where $Z_{j}^{2n}(R)=\{z\in\mathbb{R}^{2n}:|x_{j}|^{2}+|p_{j}|^{2}\leq
R^{2}\}$.

The existence of symplectic capacities is guaranteed Gromov's \ non-squeezing
theorem which says that there exists $f\in\operatorname*{Symp}(n)$ such that
$f(B^{2n}(R))\subset Z_{j}^{2n}(r)$ if and only if $R\leq r$. It follows from
this deep result that the mappings $c_{\min}$ and $c_{\max}$ defined by
\begin{subequations}
\begin{align}
c_{\min}(\Omega)  &  =\sup_{f\in\operatorname*{Symp}(n)}\{\pi R^{2}%
:f(B^{2n}(R))\subset\Omega\}\label{cmin}\\
c_{\max}(\Omega)  &  =\inf_{f\in\operatorname*{Symp}(n)}\{\pi R^{2}%
:f(\Omega)\subset Z_{j}^{2n}(R) \label{cmax}%
\end{align}
are symplectic capacities ($c_{\min}$ is sometimes called \textquotedblleft
Gromov's width\textquotedblright) and that every symplectic capacity $c$
satisfies $c_{\min}$ $\leq c\leq c_{\max}$ (for surveys of symplectic
capacities from an elementary point of view see \cite{golu09,Birk,Birkbis}).

Symplectic capacities are generally notoriously difficult to calculate, even
numerically; \ see the project website \cite{ETH}. The case of ellipsoids
\end{subequations}
\[
\Omega=\{z\in\mathbb{R}^{2n}:Mz\cdot z\leq\hbar
\]
where $M\in\operatorname*{Sym}2n,\mathbb{R})$, $M>0$ is however well-known:
for every symplectic capacity $c$ we have
\begin{equation}
c(\Omega)=\frac{\pi\hbar}{\lambda_{\max}^{\sigma}} \label{capellipsoid}%
\end{equation}
where $\lambda_{\max}^{\sigma}$ is the largest symplectic eigenvalue of the
matrix $M$ (the symplectic eigenvalues of $M$ are the numbers $\lambda
_{j}^{\sigma}>0$ ($1\leq j\leq n$) such that the $\pm i\lambda_{j}^{\sigma}$
are the eigenvalues of $JM$). .A very useful symplectic capacity is that
defined by Hofer and Zehnder capacity \cite{HZ} which has the property that
when $\Omega$ is a compact convex set in $(\mathbb{R}^{2n},\sigma)$ with
smooth boundary $\partial\Omega$ then
\begin{equation}
c_{\mathrm{HZ}\ }(\Omega)=\int_{\gamma_{\min}}pdx=\sum_{j=1}^{n}\int%
_{\gamma_{\min}}p_{j}dx_{j} \label{HZ}%
\end{equation}
where $\gamma_{\min}$ is the shortest (positively oriented) Hamiltonian
periodic orbit carried by $\partial\Omega$.

\subsection{The symplectic capacity of a quasi state}

The following result generalizes the discussion of the Mahler volume using
symplectic capacities.

\begin{theorem}
Let $(X_{\ell}\times X_{\ell^{\prime}}^{\hbar})=S(X\times X^{\hbar})$,
$S\in\operatorname*{Sp}(n)$, be a pure quasi state. We have%
\begin{equation}
c_{\max}(X_{\ell}\times X_{\ell^{\prime}}^{\hbar})=4\hbar. \label{clh}%
\end{equation}
For a general quasi state \ $(X_{\ell}\times P_{\ell^{\prime}})$,
$X_{\ell^{\prime}}^{\hbar}\subset P_{\ell^{\prime}}$ we have
\begin{equation}
c_{\max}(X\times P)=4\lambda_{\max}\hbar\label{yaron1}%
\end{equation}
where $\lambda_{\max}\geq1$ is the number
\begin{equation}
\lambda_{\max}=\max\{\lambda>0:\lambda X^{\hbar}\subset P\}.
\end{equation}

\end{theorem}

\begin{proof}
For a detailed proof when $X$ is an ellipsoid see \cite{gopolar}; also Prop. 3
in \cite{ACHAPOLAR}). In the general case one has to use results Ii
\cite{Artstein} Artstein-Avidan \textit{et al}. show that for $\hbar=1$ we
have $c_{\max}(X\times X^{1})=4$. Using the obvious relation $X=\hbar X^{1}$
we have, by the conformality property of symplectic capacities,%
\begin{align*}
c_{\max}(X\times X^{\hbar})  &  =c_{\max}(\hbar^{1/2}(\hbar^{-1/2}X\times
\hbar^{1/2}X^{1}))\\
&  =\hbar c_{\max}(\hbar^{-1/2}X\times\hbar^{1/2}X^{1})\\
&  =\hbar c_{\max}(X\times X^{\hbar})
\end{align*}
the last equality because $\hbar^{-1/2}X\times\hbar^{1/2}X^{1}=M_{\hbar^{1/2}%
}(X,X^{1})$ with $M_{\hbar^{1/2}}\in\operatorname*{Sp}(n)$; hence $c_{\max
}(X\times X^{\hbar})=$ $4\hbar$. Formula (\ref{clh}) follows since by the
symplectic invariance of symplectic capacities we have
\[
c_{\max}(X_{\ell}\times X_{\ell^{\prime}}^{\hbar})=c_{\max}((X\times X^{\hbar
}))=c_{\max}(X\times X^{\hbar}).
\]
A similar argument using the formula $c_{\max}(X\times P)=4\lambda_{\max}$ in
\cite{Artstein} leads to (\ref{yaron1}).
\end{proof}

\subsection{The covariance matrix and quantum blobs}

Let us begin by studying how our quasi states are related to her usual
Heisenberg inequality and its generalizations. Let $|\psi\rangle$ be a quantum
taste, we assume \ that $\psi$ has first and second momenta so the state has a
center and covariance matrix $\Sigma$; \ we denote by $\Omega_{\Sigma}$ the
corresponding covariance ellipsoid:%
\[
\Sigma=\{z:\frac{1}{2}\Sigma^{-1}z\cdot z\leq1\}.
\]
We will use the block-matrix notation%
\begin{equation}
\Sigma=%
\begin{pmatrix}
\Sigma_{XX} & \Sigma_{XP}\\
\Sigma_{PX} & \Sigma_{PP}%
\end{pmatrix}
\text{ \ },\text{ \ }\Sigma_{PX}=\Sigma_{XP}^{T} \label{defcovma}%
\end{equation}
with $\Sigma_{XX}=(\Delta(x_{j},x_{k}))_{1\leq j,k\leq n}$, \ $\Delta
(x_{j},x_{j})=\Delta x_{j},^{2}$, and. so on. The following result
\ \cite{dutta,Birk,Birkbis} is well-known: the covariance matrix of a quantum
state (pure or mixed= satisfies%
\begin{equation}
\Sigma+\frac{i\hbar}{2}J\text{ \emph{is semi-definite positive}}\emph{.}
\label{Quantum}%
\end{equation}
Condition (\ref{Quantum}), which we will write for short $\Sigma+\frac{i\hbar
}{2}J\geq0$, is also necessary for a density operator $\widehat{\rho}$ to to
be positive semi-definite (and hence represent a mixed quantum state)
\ \emph{but it \ is not sufficient e}xcept in the Gaussian case, see
\cite{Birkbis,QHA} and the references therein. \ Condition (\ref{Quantum}) is
an algebraic way to represent the Robertson--Schr\"{o}dinger inequalities for
the (co)variances:%
\begin{equation}
\Delta(x_{j},,x_{j})\Delta(p_{j},p_{j})\geq\Delta(x_{j},p_{j})^{2}+\tfrac
{1}{4}\hbar^{2}\text{ \ },\text{ \ }1\leq j\leq n. \label{RS}%
\end{equation}
It turns out that we have showed in \cite{FOOP,golu09,Birkbis} that

\begin{proposition}
The condition $\Sigma+\frac{i\hbar}{2}J\geq0$ is equivalent to the following:
the covariance ellipsoid
\[
\Omega_{\Sigma}=\{z:\tfrac{1}{2}\Sigma^{-1}z\cdot z\leq1\}
\]
contains a quantum blob.
\end{proposition}

Notice that this result in particular implies that the symplectic capacity
cannot be arbitrarily small:

\begin{corollary}
(i) The covariance ellipsoid of a granum state $\Omega_{\Sigma}$ satisfies
$c(\Omega_{\Sigma})\geq\pi\hslash$ for every symplectic capacity $C(ii)$ (ii)
; the shortest periodic Hamiltonian orbit $\gamma_{\min}$ carried by satisfies
$\Omega_{\Sigma}$ satisfies
\begin{equation}
\int_{\gamma_{\min}}pdx\geq\frac{1}{2}h. \label{periodo}%
\end{equation}

\end{corollary}

\begin{proof}
(i) Since $\Omega_{\Sigma}$ contains a quantum blob $Q_{S}=S(B^{2n}%
(\sqrt{\hbar}))$, \ we have, by the monotonicity property of symplectic
capacities,
\[
c(\Omega_{\Sigma})\geq(cQ_{S})=c(B^{2n}(\sqrt{\hbar}))=\pi\hbar.
\]
(ii) Since the inequality $c(\Omega_{\Sigma})\geq\pi\hslash$ hold for every
symplectic capacity, it holds for the Hofer--Zehnder capacity $c_{\mathrm{HZ}%
\ }.$ Formula (\ref{periodo}) follows from (\ref{HZ}) since $\Omega_{\Sigma}$
is convex.
\end{proof}

\section{ Gaussian Quasi States}

We now focus on the case where $X$ (and hence $X^{\hbar})$ is an ellipsoid.
This leads us to the case where the involved quasp states are Gaussian, in a
sense that will be explained.

\subsection{Quantum blobs and pre-Iwasawa factorization}

We introduced the notion of quantum blob is in \cite{blob,physletta,golu09} as
being a symplectic ball with radius $\sqrt{\hbar}$ that is an ellipsoid of
$\mathbb{R}^{2n}$ of the type
\begin{equation}
Q_{S}(z_{0})=T(z_{0})S(B^{2n}(\sqrt{\hbar}))=S(B^{2n}(S^{-1}z_{0},\sqrt{\hbar
}) \label{blob}%
\end{equation}
where $S\in\operatorname*{Sp}(n)$, $T(z_{0}):z\longmapsto z+z_{0}$, and
$B^{2n}(z_{0}\sqrt{\hbar})$ \ is the ball in $\mathbb{R}^{2n}$ with radius
$\sqrt{\hbar}$ centered at $z_{0}$. \ We will write for short $Q_{S}(0)=Q_{S}%
$. The inhomogeneous symplectic group $\operatorname{ISp}%
(n)=\operatorname*{Sp}(n)\rtimes\mathbb{R}^{2n}$ obviously acts transitively
on the set $\operatorname*{Blob}(n)$ of all quantum blobs:%
\begin{gather*}
\operatorname{ISp}(n)\times\operatorname*{Blob}(n)\longrightarrow
\operatorname*{Blob}(n)\\
(S^{\prime}T(z^{\prime}),Q_{S}(z))\longmapsto Q_{S^{\prime}S}(S^{\prime
}(z^{\prime}+z)).
\end{gather*}

Recall the pre-Iwasawa decomposition of a symplectic matrix\cite{dutta,iwa}:
there exist unique matrices $P,L\in\operatorname*{Sym}(n,\mathbb{R})$, $L>0$,
and $R\in\operatorname*{Sp}(n)\cap O(2n)$ such that
\begin{equation}
S=V_{P}M_{L}R. \label{iwa1}%
\end{equation}
These matrices are explicitly given by the formulas%
\begin{gather}
L=(AA^{T}+BB^{T})^{-1/2}\label{pl1}\\
P=-(CA^{T}+DB^{T})(AA^{T}+BB^{T})^{-1}. \label{pl2}%
\end{gather}
The matrix $R$ is a symplectic rotation: writing
\[
R=%
\begin{pmatrix}
E & F\\
-F & E
\end{pmatrix}
\in\operatorname*{Sp}(n)\cap O(2n,\mathbb{R})
\]
the $n\times n$ blocks $E$ and $F$ are given by%
\begin{equation}
E=(AA^{T}+BB^{T})^{-1/2}A\text{ \ },\text{ \ }F=(AA^{T}+BB^{T})^{-1/2}B.
\label{unixy}%
\end{equation}

\begin{proposition}
\label{Prop1}Every quantum blob $Q_{S}(z)$ can be written (in a unique way)
\begin{equation}
Q_{S}(z)=T(z)V_{P}M_{L}(B^{2n}(\sqrt{\hbar}))=T(z)M_{L}V_{LPL^{-1}}%
(B^{2n}(\sqrt{\hbar})) \label{mele}%
\end{equation}
for some $P,L\in\operatorname*{Sym}(n,\mathbb{R})$, $L>0$, where \ $P$ and $L$
correspond to the pre-Iwasawa factorization of $S$.
\end{proposition}

\begin{proof}
Obvious since $R((B^{2n}(\sqrt{\hbar})=B^{2n}(\sqrt{\hbar})$. The relation
$V_{P}M_{L}=M_{L}V_{LPL^{-1}}$ is trivial since $L=L^{T}$.
\end{proof}

The phase space ball $B^{2n}(\sqrt{\hbar})$ is the simplest example of a
quantum blob. We have just seen that $R(B^{2n}(\sqrt{\hbar}))=B^{2n}%
(\sqrt{\hbar})$ for every $R\in\operatorname*{Sp}(n)\cap O(2n,\mathbb{R})$.
Assume conversely that $S(B^{2n}(\sqrt{\hbar})=B^{2n}(\sqrt{\hbar}$ for some
$S\in\operatorname*{Sp}(n)$. Then, using the pre-Iwasawa factorization of $S$
this condition is equivalent to $V_{P}M_{L}(B^{2n}(\sqrt{\hbar}))=B^{2n}%
(\sqrt{\hbar})$, which is only possible if $P=0$ and $L=I$. \ 

\subsection{Generalized coherent states}

We consider generalized Gaussians (\textquotedblleft coherent
states\textquotedblright) of the type
\begin{equation}
\psi_{X,Y}(x)=\left(  \tfrac{\det X}{(\pi\hbar)^{n}}\right)  ^{1/4}%
e^{-\frac{1}{2\hbar}(X+iY)x\cdot x} \label{squeezed}%
\end{equation}
where $X,Y\in\operatorname*{Sym}(n,\mathbb{R})$ ,and $X$ is positive definite
($X>0$). These Gaussians can be obtained from the standard coherent state
\begin{equation}
\phi_{0}^{\hbar}(x)=\psi_{I,0}(x)=(\pi\hbar)^{-n/4}e^{-|x|^{2}/2\hbar}
\label{standard}%
\end{equation}
using elementary metaplectic operators, as follows from the obvious formula
\begin{equation}
\psi_{X,Y}=\widehat{S}_{X,Y}\phi_{0}^{\hbar}=\widehat{V}_{Y}\widehat{M}%
_{X^{1/2}}\phi_{0}^{\hbar} \label{fixy}%
\end{equation}
where $\widehat{V}_{Y}$ and $\widehat{M}_{X^{1/2}}=\widehat{M}_{X^{1/2},0}$
are defined by (\ref{mp2}) and (\ref{mp3}). Beware: neither the operators
$\widehat{S}_{X,Y}n$ nor their projections $S_{X,Y}$ \ form a group. This is
due to the fact that even if $X$ and $X^{\prime}$ are symmetric,
$X(X\prime)^{-1}$ is not.)

More generally, we define the displaced Gaussians
\[
\psi_{X,Y}^{z_{0}}=\widehat{T}(z_{0})\psi_{X,Y}\text{ \ \ },\text{ \ \ }%
z_{0}=(x_{0},p_{0})
\]
to which the transformations above are easily extended. We denote by
$\operatorname*{Gauss}(n)$ the set of all \textquotedblleft Gaussian
states\textquotedblright\ $|\psi_{X,Y}^{z_{0}}\rangle$ (\textit{i.e.} the
collection of all functions $c\psi_{X,Y}^{z_{0}}$where $c\in\mathbb{C}$ with
\ $|c|=1$). The subset of $\operatorname*{Gauss}(n)$ consisting of all
$|\psi_{X,Y}\rangle$ is denoted by $\operatorname*{Gauss}_{0}(n)$. The
following result identifies $\operatorname*{Blob}(n$ and
$\operatorname*{Gauss}(n)$.

\begin{theorem}
\label{Prop2}The mapping
\[
\Gamma:\operatorname*{Blob}(n)\longrightarrow\operatorname*{Gauss}(n)
\]
defined by
\[
\Gamma:T(z_{0})S_{X,Y}B^{2n}(\sqrt{\hbar}))\longmapsto|\psi_{X,Y}^{z_{0}%
}\rangle=|\widehat{T}(z_{0})\widehat{S}_{X,Y}\phi_{0}^{\hbar}\rangle
\]
where $S_{X,Y}=V_{-Y}$ $M_{X^{-1/2}}$ and $\widehat{S}_{X,Y}=\widehat{V}%
_{Y}\widehat{M}_{X^{-1/2}}$ is a bijection such that
\[
\Gamma(S^{\prime}Q(S,z_{0})=\Gamma(|\widehat{S^{\prime}}\psi_{X,Y}^{z_{0}%
}\rangle
\]
for all $S,S^{\prime}\in\operatorname*{Sp}(n)$ if $\widehat{S^{\prime}}$
covers $S^{\prime}$.
\end{theorem}

\begin{proof}
It suffices to consider the case $z=0$. and to use the intertwining formulas
\[
\widehat{S}\widehat{T}(z)=\widehat{T}(Sz)\widehat{S}\text{ \ },\text{
\ }ST(z)=T(Sz)S
\]
to deal with the general case. Let us thus show that the restriction
\begin{gather*}
\Gamma_{0}:\operatorname*{Blob}\nolimits_{0}(n)\longrightarrow
\operatorname*{Gauss}\nolimits_{0}(n)\\
S_{X,Y}B^{2n}(\sqrt{\hbar})\longmapsto|\widehat{S}_{X,Y}\phi_{0}^{\hbar
}\rangle
\end{gather*}
is a bijection. Firstly, $\Gamma_{0}$ is a well-defined mapping since every
quantum blob $Q_{S}$ can be written using the pre-Iwasawa factorization as
(\ref{iwa1})) as%
\[
Q_{S}=S_{X,Y}(B^{2n}(\sqrt{\hbar})=V_{Y}M_{X^{-1/2}}(B^{2n}(\sqrt{\hbar}).
\]
Similarly every Gaussian function $\psi_{X,Y}$ can be written as $\psi
_{X,Y}=\widehat{S}_{X,Y}\phi_{0}^{\hbar}$, showing at the same time that
$\Gamma_{0}$ is surjective. To show that $\Gamma_{0}$ is \ bijection there
remains to \ prove injectivity, that is if $\widehat{S}_{X,Y}\phi_{0}^{\hbar
}=\widehat{S}_{X^{\prime},Y^{\prime}}^{\prime}\phi_{0}^{\hbar}$ then
$S_{X,Y}B^{2n}(\sqrt{\hbar}))=S_{X^{\prime},Y^{\prime}}B^{2n}(\sqrt{\hbar}))$.
In view of the rotational symmetry of the standard coherent state $\phi
_{0}^{\hbar}$ we must have $\widehat{S}_{X,Y}=\widehat{S}_{X^{\prime
},Y^{\prime}}\widehat{R}$ where $\widehat{R}\in\operatorname*{Mp}(n)$ \ covers
a symplectic rotation $R\in\operatorname*{Sp}(n)\cap O(2n,\mathbb{R})$, hence
$S_{X^{\prime},Y^{\prime}}=S_{X,Y}R$ and the injectivity follows since
$R(B^{2n}(\sqrt{\hbar})))=B^{2n}(\sqrt{\hbar}))$.
\end{proof}

\subsection{Gaussian quasi states and the Wigner transform}

There is another way,to recover the results above, using the Wigner transform.
We have \cite{Birk,BBirkbis,,WIGNER}
\[
W\phi_{0}^{\hbar}(z)=(\pi\hbar)^{-n}e^{-|z|^{2}/\hbar}%
\]
and hence, by the symplectic covariance of the Wigner transform%
\[
W\psi_{X,Y}(z)=W(\widehat{S}_{X,Y}\phi_{0}^{\hbar})(z)=W\phi_{0}^{\hbar
}(S_{X,Y}^{--1}z)
\]
which yields, since $\widehat{S}_{X,Y}=\widehat{V}_{Y}\widehat{M}_{X^{1/2}}$
has projection $S_{X,Y}=V_{Y}M_{X^{1/2}}$ on $\operatorname*{Sp}(n)$,
\begin{equation}
W\psi_{X,Y}(z)=(\pi\hbar)^{-n}e^{--\frac{1}{\hbar}G_{X,Y}z.z}\text{ }
\label{wpsi}%
\end{equation}
where $G_{X,Y}\in\operatorname*{Sp}(n)$ is the symmetric positive definite
matrix%
\[
G_{X,Y}=(S_{X,Y}S_{X,Y}^{T})^{-1}=%
\begin{pmatrix}
X+YX^{-1}Y & YX^{-1}\\
X^{-1}Y & X^{-1}%
\end{pmatrix}
.
\]
The Wigner transform $W\psi_{X,Y}$ is thus a centered Gaussian probability
distribution with covariance matrix $\Sigma=\frac{\hbar}{2}$ $G_{X,Y}^{-1}$.
The associated covariance ellipsoid is the quantum blob defend by
$G_{X,Y}z\cdot z\leq\hbar$:
\begin{equation}
\Omega_{\Sigma}=\{z:\tfrac{1}{2}\Sigma^{-1}z\cdot z\leq1\}=S_{X,Y}%
(B^{2n}(\sqrt{\hbar}). \label{Wigell}%
\end{equation}
Let us consider, more generally a Gaussian distribution%
\begin{equation}
W_{\widehat{\rho}}(z)=\left(  \tfrac{1}{2\pi}\right)  ^{n}(\det\Sigma
)^{-1/2}e^{-\frac{1}{2}\Sigma^{-1}z\cdot z} \label{wigrho}%
\end{equation}
whose purity is \cite{Birk}
\begin{equation}
\mu(\widehat{\rho})=\operatorname*{Tr}(\widehat{\rho}^{2})=\left(  \frac
{\hbar}{2}\right)  ^{n}(\det\Sigma)^{-1/2}. \label{purity}%
\end{equation}
\ 

\section{The Canonical Group of a Quasi State}

Here we ask \textquotedblleft what is the subgroup of $\operatorname*{Sp}(n)$
leaving an arbitrary (centered) quantum blob invariant? A characteristic
property of quasi-states arising from elliptic set $X_{\ell}$ carried by a
Lagrangian plane $\ell$ is tat they contain a quantum blob, \ to which a
Gaussian state is canonically associated. We now address the question of
whether several different quasi states can give rise to the same Gaussian.
Technically, this amounts to analyzing the following situation:

\subsection{A stationary Schr\"{o}dinger equation for $\psi_{X,Y}$}

The standard Gaussian $\phi_{0}^{\hbar}$ trivially satisfies the partial
differential equation%
\begin{equation}
\widehat{H}_{0}\phi_{0}^{\hbar}=\frac{1}{2}(-\hbar^{2}\nabla_{x}^{2}%
+|x|^{2})\phi_{0}^{\hbar}=\frac{1}{2}n\hbar\phi_{0}^{\hbar} \label{Ho}%
\end{equation}
(it is the stationary Schr\"{o}dinger equation for the ground state of the
isotropic harmonic oscillator with mass $m=1$). More generally we would now
like to find an equation $\widehat{H}_{X,Y}\psi_{X,Y}=\lambda\psi_{X,Y}$
satisfied by $\psi_{X,Y}$. For this purpose it is tempting to use the relation
(\ref{fixy}) earlier established and to rewrite (\ref{Ho}) noting that
$\phi_{0}^{\hbar}=\widehat{S}_{X,Y}^{-1}\psi_{X,Y}$. This leads to the
equation
\[
\widehat{H}_{X,Y}\psi_{X,Y}=\widehat{S}_{X,Y}\widehat{H}_{0}\widehat{S}%
_{X,Y}^{-1}\psi_{X,Y}=\frac{1}{2}n\hbar\psi_{X,Y}%
\]
or, equivalently if we use Weyl quantization \cite{Birk,Birkbis}%
\[
\widehat{H}_{X,Y}\psi_{X,Y}=\widehat{H_{00}\circ S_{X,Y}^{-1}}\psi_{X,Y}%
=\frac{1}{2}n\hbar\psi_{X,Y}%
\]
where $H_{0}(x,p)=\frac{1}{2}(|x|^{2}+|p|^{2})$. \ However, this is just an
equivalent way to rewrite the equation (\ref{Ho}); and does not bring any new
information.. We will rather proceed by making a direct computation (which is
a particular case of \textquotedblleft Fermi's trick\textquotedblright\ we
will discuss in a moment). Consider first the real Gaussian $\psi_{X}%
=\psi_{X,0}$. Diagonalizing $X$ an easy calculation leads to the equation%
\begin{equation}
\widehat{H}_{X,0}\psi_{X}=\frac{1}{2}(-\hbar^{2}\nabla_{x}^{2}+|X^{2}x\cdot
x)\psi_{X0}=\frac{1}{2}\hbar\operatorname*{Tr}(X)\psi_{X} \label{H1}%
\end{equation}
where $\operatorname*{Tr}(X)$ is the trace of the symmetric matrix $X$. \ We
thereafter notice that $\psi_{X,Y}$ is obtained from $\psi_{X}$ by multiplying
it by the complex exponential $V_{Y}=e^{-iYx\cdot x/2\hbar}$ and this amounts
to making the change of gauge $-i\hbar\nabla_{x}\longrightarrow-i\hbar
\nabla_{x}+Yx$ in (\ref{H1}), which leads to the equation \
\begin{equation}
\widehat{H}_{X,Y}\psi_{X,Y}=\frac{1}{2}\hbar\operatorname*{Tr}(X)\psi_{X,Y}
\label{H2}%
\end{equation}
where $\widehat{H}_{X,Y}$ \ is the second order partial differential operator%
\begin{equation}
\widehat{H}_{X,Y}=\frac{1}{2}(-i\hbar\nabla_{x}+Yx)^{2}+X^{2}x\cdot x.
\label{hhat}%
\end{equation}
Notice that$\widehat{H}_{X,Y}$ is the Weyl quantization (or any other
\textquotedblleft reasonable\textquotedblright, \textit{i.e.} \ the
Born--Jordan quantization) of the quadratic polynomial
\begin{equation}
H_{X,Y}(x,p)=\frac{1}{2}\left(  (p+Yx)^{2}+X^{2}x\cdot x\right)  \label{xyh}%
\end{equation}
which we call (for reasons that will become clear later on) the\emph{ Fermi
Hamiltonian of }$\psi_{X,Y}$. The latter can be written in matrix form as
\begin{equation}
H_{X,Y}(z)=\frac{1}{2}M_{X,Y}z\cdot z \label{gf4}%
\end{equation}
where $M_{X,Y}$ is the symmetric and positive-definite matrix%
\begin{equation}
M_{X,Y}=%
\begin{pmatrix}
X^{2}+Y^{2} & Y\\
Y & I
\end{pmatrix}
. \label{mxy}%
\end{equation}
The latter factorizes as
\begin{equation}
M_{X,Y}=(S_{X,Y}^{-1})^{T}D_{X}S_{X,Y}^{-1}\text{ \ },\text{\ \ \ }D_{X}=%
\begin{pmatrix}
X & 0\\
0 & X
\end{pmatrix}
\label{mfs}%
\end{equation}
where $S_{X,Y}\in\operatorname*{Sp}(n)$ is given by%
\begin{equation}
S_{X,Y}=V_{Y}M_{X^{1/2}}=%
\begin{pmatrix}
X^{-1/2} & 0\\
-YX^{-1/2} & X^{-1/2}%
\end{pmatrix}
. \label{ess}%
\end{equation}
Note that
\[
H_{X,Y}\circ S_{X,Y}(z)=\frac{1}{2}D_{X}z\cdot z=\frac{1}{2}Xx\cdot x+\frac
{1}{2}Xp\cdot p.
\]
The flow generated by this the Hamiltonian function $K_{X}=H_{X,Y}\circ
S_{X,Y}$ thus consists of symplectic rotations $R_{t}=e^{tJD_{X}}=e^{tD_{X}J}%
$, explicitly given by
\[
R_{t}=%
\begin{pmatrix}
\cos(tX) & \sin(tX)\\
-\sin(tX) & \cos(tX)
\end{pmatrix}
\in\operatorname*{Sp}(n)\cap O(2n,\mathbb{R}).
\]

The way we obtained the result above is a actually a particular case of a more
general (and simple ) procedure outlined by Enrico Fermi \cite{Fermi} back in
1930 (also see \cite{ben} who rediscovered Fermi's construction). \ We have
studied the Fermi Hamiltonian \cite{DeGoHi} from the point perspective of an
internal energy associated with Bohm's theory of quantum motion.

\subsection{Construction of the canonical group}

Consider now the Hamilton equations for the Fermi Hamiltonian $H_{X,Y}$; they
can be written%

\[
\dot{z}(t)=J\nabla_{z}H_{X,Y}(z(t))=2JM_{X,Y}z(t).
\]
It follows \ from the factorization (\ref{mfs}) and the relation
$J(S_{X,Y}^{-1})^{T}=S_{X,Y}$ (since $S_{X,Y}\in\operatorname*{Sp}(n)$) that
the Hamiltonian flow determined by $H_{X,Y}$ consists of the symplectic
matrices%
\[
S_{t}=e^{2tJM_{X,Y}}=S_{X,Y}e^{2JD_{X}}S_{X,Y}^{-1}.
\]
We will call the one-parameter subgroup $(S_{t})$ of $\operatorname*{Sp}(n)$
the \emph{canonical group of} $\psi_{X,Y}$ (which we introduces in
\cite{jgeom} in another the context). It is clear, by conservation of energy,
that $(S_{t})$ leaves the Fermi ellipsoid $\Omega_{X,Y}$ invariant:
$S_{t}(\Omega_{X,Y})=\Omega_{X,,Y}$ for all $t$. More interesting is the
following result, which makes use of the fact that an one-parameter subgroup
of a connected Lie group can be covered by a unique parameter subgroup of each
of its coverings (see \textit{e.g.} Steenrod \cite{Steenrod}.)

\begin{theorem}
(i) The lift $(\widehat{S}_{t})$ of the canonical group $\ (S_{t})$\ \ to the
metaplectic group $\operatorname*{Mp}(n)$preserves the state $|\psi
_{X,Y}\rangle$ invariant; in fact%
\begin{equation}
\widehat{S}_{t}\psi_{X,Y}=e^{-it\operatorname*{Tr}/X)}\psi_{X,Y}.
\label{stpsi}%
\end{equation}
(ii) The canonical group $(S_{t})$ leaves the quantum blob $S_{X,Y}%
B^{2n}(\sqrt{\hbar})$ invariant:
\begin{equation}
S_{t}(S_{X,Y}B^{2n}(\sqrt{\hbar}))=S_{X,Y}B^{2n}(\sqrt{\hbar}) \label{stblob}%
\end{equation}
for all $t$. \ 
\end{theorem}

\begin{proof}
(i) Consider the time-dependent Schr\"{o}dinger equation%
\begin{equation}
i\hbar\frac{\partial\psi}{\partial t}=\widehat{H}_{X,Y}\psi\text{ \ , \ }%
\psi(\cdot,0)=\psi_{X,Y} \label{schr1}%
\end{equation}
and let $(\widehat{S}_{t})$ \ be defined as above. We have, according to the
theory of the Schr\"{o}dinger equation for quadratic Hamiltonians
\cite{ICP,Birk,Birkbis} the unique solution of (\ref{schr1}) is given by
\[
\psi(x,t)=\widehat{S}_{t}\psi_{X,Y}(x).
\]
On the other hand, solving (\ref{schr1}) by the method of separation of
variables immediately yields, taking formula (\ref{H2}) into account%
\[
\psi(x,t)=e^{-it\operatorname*{Tr}/X)}\psi_{X,Y};
\]
comparing both solutions yields (\ref{stpsi}). (ii) Immediately follows from
Theorem \ref{Prop2}: the mapping
\begin{gather*}
\Gamma_{0}:\operatorname*{Blob}\nolimits_{0}(n)\longrightarrow
\operatorname*{Gauss}\nolimits_{0}(n)\\
S_{X,Y}B^{2n}(\sqrt{\hbar})\longmapsto|\psi_{X,Y}\rangle
\end{gather*}
is a a bijection and we have $|\psi_{X,Y}\rangle=|\widehat{S}_{t}\psi
_{X,Y}\rangle$.
\end{proof}

\begin{acknowledgement}
This work has been financed by the QUANTUM\ AUSTRIA grant PAT 2056623 of the
Austrain Research Foundation FWF.
\end{acknowledgement}

\end{document}